\documentclass[letterpaper, 10 pt, conference]{ieeeconf}  

\IEEEoverridecommandlockouts                              

\usepackage{times} 
\usepackage{amsmath} 
\usepackage{amssymb}  
\usepackage{mathtools}
\usepackage{tikz}
\usetikzlibrary{patterns,decorations.pathreplacing}
\usepackage{pgfplots}
\pgfplotsset{compat=1.16}

\usepackage{xcolor}
\newtheorem{lemma}{\bf Lemma}
\newtheorem{theorem}{\bf Theorem}
\newtheorem{corollary}{\bf Corollary}

\newtheorem{assumption}{\bf Assumption}

\usepackage{stfloats}
\usepackage{float}

\usepackage{calc}

\usepackage{acronym}
\acrodef{QMI}{Quadratic Matrix Inequality}
\acrodef{LTI}{linear time-invariant}
\acrodef{LFT}{linear fractional transformation}
\acrodef{SDP}{semidefinite program}
\acrodef{IQC}{integral quadratic constraint}

\usepackage{cite}

\title{A System Parameterization for Direct Data-Driven Estimator Synthesis}

\author{Felix Brändle and Frank Allgöwer 
	\thanks{F. Allg\"ower is thankful that this work was funded by the Deutsche Forschungsgemeinschaft (DFG, German Research Foundation) under Germany’s Excellence Strategy -- EXC 2075 -- 390740016 and within grant AL
		316/15-1 – 468094890. F. Br\"andle thanks the International Max Planck Research School for Intelligent Systems (IMPRS-IS) for supporting him.
	}
	\thanks{F. Br\"andle and F. Allg\"ower are with the University of Stuttgart, Institute for Systems
		Theory and Automatic Control, 70550 Stuttgart,
		Germany. (e-mail: \{felix.braendle, frank.allgower\}@ist.uni-stuttgart.de)}%
}

\begin{document}

\pubid{\begin{minipage}{\textwidth}\ \\[60pt] \copyright 2025 IEEE. This version has been accepted for publication in IEEE Control Systems Letters Volume 9. Personal use of this material is permitted. Permission from IEEE must be obtained for all other uses, in any current or future media, including reprinting/republishing this material for advertising or promotional purposes, creating new collective works, for resale or redistribution to servers or lists, or reuse of any copyrighted component of this work in other works.\end{minipage}}

\maketitle

\begin{abstract}
This paper introduces a novel parameterization to characterize unknown \acl{LTI} systems using noisy data.
The presented parameterization describes exactly the set of all systems consistent with the available data.
We then derive verifiable conditions when the consistency constraint reduces the set to the true system and when it does not have any impact.
Furthermore, we demonstrate how to use this parameterization to perform a direct data-driven estimator synthesis with guarantees on the $\mathcal{H}_\infty$-norm.
Lastly, we conduct numerical experiments to compare our approach to existing methods.
\end{abstract}


\section{INTRODUCTION}
In recent years, interest in system analysis and controller design based on collected data has been increasing \cite{Coulson2019, Martin2023a}.
A common approach involves first identifying a model from data and then employing model-based techniques to synthesize a controller or verify system properties.
Typically, only a finite number of noisy data points are available, making identifying the true system challenging \cite{Oymak2019}.
The difference between true and identified system may lead to instabilities when the identified model is used for controller synthesis.
Direct data-driven control methods aim to address this problem by not identifying a single model, but instead parameterizing a set of systems using the available data. 
This data-dependent set is constructed to contain the true system.
By using robust control techniques, it is possible to give guarantees for each element of the set, and hence also for the true system \cite{Berberich2019}.
However, the complexity of these set-membership techniques grows rapidly with the number of samples \cite{Castano2011} and may require computationally intensive methods like sum-of-squares \cite{Miller22}.
So the overall goal is to find a tight parameterization that can also be used with efficient synthesis algorithms.

A commonly used approach is based on the work in \cite{DePersis2020}.
This approach parameterizes the set using a prior known bound on the noise and a right inverse of the data matrices.
However, this approach overapproximates the set by not including the additional knowledge that the noise must also satisfy the model equation \cite{Berberich2019, Koch2020}.
Results based on the data-informativity framework \cite{Waarde2020a, Waarde2022} are particularly attractive, because they exactly parameterize all systems consistent with the observed data and a known prior bound on the noise, leading to a less conservative set description.
Using robust control techniques, it is possible to synthesize a guaranteed stabilizing state feedback controller or verify dissipativity properties.
Another difference between these approaches is whether the parameterization is derived for the primal or the dual space.
The data-informativity framework uses a \ac{QMI} with the transposed parameter matrix leading to a description in the dual space. 
In contrast, the approach based on the right inverse provides a direct parameterization in the primal space without transposing the parameter matrix.
The dual space is particularly suited for robust system analysis and robust state feedback synthesis \cite[Section 8.3]{Weiland1994}, as it leads to a convex \ac{SDP}.
Meanwhile, the primal space is well suited for solving the robust estimator synthesis, robust system analysis, and \acp{IQC} \cite{Holicki2023a}, since these problems can be convexified.
However, robust state feedback synthesis does not lead to a convex problem in primal space and vice versa for the estimator synthesis in dual space, such that each space is better suited to solve different control problems.

One approach to handle the difference between dual and primal is described in \cite{Mishra2022}, which employs the Dualization Lemma \cite[Lemma 4.9]{Weiland1994} to transform equivalently between dual and primal space in order to synthesize a robust estimator in primal space. 
However, this method requires the system matrices to be known, making this approach not suitable in many setups. 
Furthermore, some additional invertibility conditions must be met to apply the Dualization Lemma.
Other approaches in the primal space necessitate noiseless data \cite{Turan2022} or involve parameter tuning to ensure stability for sufficiently small noise during the data collection phase \cite{Liu2023a}.

In this work, we derive an explicit and consistent system parameterization for a given dataset and employ it to synthesize an estimator with a guaranteed upper bound on the closed loop $\mathcal{H}_\infty$-norm even for a finite number of noisy data points.
Our approach extends the data-informativity framework to the primal space, without requiring additional invertibility assumptions to apply dualization.
Furthermore, we compare our approach to the one based on the right inverse and derive conditions when both approaches are equal to bridge the gap between the right inverse approach and the data-informativity framework.
We demonstrate the usefulness of the proposed parameterization as a analysis tool by deriving verifiable conditions on the data, when consistency leads to exact identification of the true system.

\section{Notation}

We denote the $n\times n$ identity matrix and the $p\times q$ zero matrix as $I_n$ and $0_{p\times q}$, respectively. 
We omit the indices, if the dimensions are clear from context. 
For the sake of space limitations, we use $\star$, if the the corresponding matrix can be inferred from symmetry. 
Moreover, we use $P\succ 0$ ($P\succeq 0$), if the symmetric matrix $P$ is positive (semi) definite. 
Negative (semi) definiteness is defined by $P\prec 0$ ($P\preceq 0$). 
For measurements $\{x_k,x_{k+1}\}_{k=0}^{N-1}$, we define the following matrices
\begin{equation*}
	X=\begin{bmatrix}
		x_0 & x_1 & \ldots & x_{N-1}
	\end{bmatrix},\quad
	X_+=\begin{bmatrix}
		x_1 & x_2 & \ldots & x_{N}
	\end{bmatrix}.
\end{equation*}
In addition, we use $A^\perp$ to denote a matrix containing a basis of the kernel of $A$ as rows or in case if $A$ has more rows than columns, then $A^\perp$ contains a basis of the kernel of $A^\top$ as columns, such that $AA^{\perp\top}=0$ or $A^{\perp\top}A=0$.
The symbol $\sigma_i(A)$ denotes the $i$'th singular value of $A$, $\sigma_{\mathrm{min}}(A)$ and $\sigma_{\mathrm{max}}(A)$ are the minimal and maximal singular values. 
Furthermore, we use $\mathrm{diag}(A_1,\ldots,A_\mathrm{k})$ to abbreviate a block diagonal matrix with $A_i$ as corresponding blocks.

\section{DATA-DRIVEN SYSTEM PARAMETRIZATION}\label{sec:Param}

In this work, we consider a linear regression model of the following form
\begin{equation}
	y = \Theta_{\mathrm{tr}} x,
\end{equation}
with regressand $y\in\mathbb{R}^p$, regressor $x\in\mathbb{R}^n$ and the true, but unknown parameter matrix $\Theta_{\mathrm{tr}}\in\mathbb{R}^{p\times n}$. Since $\Theta_{\mathrm{tr}}$ is unknown, we assume to have access to noisy data $\{x_k,y_k\}_{k=0}^{N-1}$, which is generated according to
\begin{equation}
	y_k = \Theta_{\mathrm{tr}} x_k + w_k,
\end{equation}
with unknown additive noise $w_k\in\mathbb{R}^p$. We rewrite the problem in matrix notation
\begin{equation}
	Y = \Theta_{\mathrm{tr}} X + W	\label{eq:param:SystemEq}
\end{equation}
with $Y\in\mathbb{R}^{p\times N}$, $X\in\mathbb{R}^{n\times N}$, and $W\in\mathbb{R}^{p\times N}$. Moreover, we assume $W$ belongs to the set $\mathcal{W}$
\begin{equation*}
	\mathcal{W}\coloneq\left\{ W\in\mathbb{R}^{p\times N}\mid
	\begin{bmatrix}W\\ I\end{bmatrix}^\top 
	\Phi
	\begin{bmatrix}	W\\ I	\end{bmatrix}\succeq 0 \right\}\label{eq:Param:W}
\end{equation*}
with known $\Phi=\Phi^\top\in\mathbb{R}^{(p+N) \times (p+N)}$ to characterize $\mathcal{W}$ by a \ac{QMI} as in \cite{Berberich2019}.
With this, we can define the set of all matrices $\Theta$ consistent with the data $\{x_k,y_k\}_{k=0}^{N-1}$ as
\begin{equation*}
	\Sigma_\Theta \coloneq \{\Theta\in\mathbb{R}^{p\times n}\mid \exists W\in\mathcal{W}:Y=\Theta X + W\}.
\end{equation*}

In the following, we denote the constraint $Y=\Theta X + W$ as consistency constraint to indicate that every element of $\Sigma_\Theta$ is able to recreate the data for some $W\in\mathcal{W}$. 
If a property holds for all $\Theta \in \Sigma_\Theta$, it is also guaranteed for the true system $\Theta_\mathrm{tr}$.
In the following, we reformulate $\Sigma_\Theta$ in terms of a single \ac{QMI} to apply robust control methods.
To do so, we first impose the following assumption on the data.
\begin{figure*}[b!]
	\noindent\makebox[\linewidth]{\rule{\textwidth}{0.4pt}}
	\begin{align}
		\tag{8}
		\begin{bmatrix}
			\Theta - \Theta_0\\ I 
		\end{bmatrix}^\top 
		\begin{bmatrix}
			-I & 0\\0 & G
		\end{bmatrix}^\top
		\left(\Phi - \Phi \begin{bmatrix}\tilde{\Theta}_0 \\ \tilde{G}\end{bmatrix} M^{-1} \begin{bmatrix}\tilde{\Theta}_0 \\ \tilde{G}\end{bmatrix}^\top \Phi \right)
		\begin{bmatrix}
			-I & 0\\0 & G
		\end{bmatrix}
		\begin{bmatrix}
			\Theta - \Theta_0\\ I 
		\end{bmatrix} \succeq 0 \label{eq:Param:TheoConsistQMI}
	\end{align}
\end{figure*}
\vspace{0.5em}

\begin{assumption} \label{ass:param:Slater}
	There exist matrices $\bar{W}\in\mathbb{R}^{p\times N}$ and $\bar{\Theta}\in\mathbb{R}^{p\times n}$ satisfying
	\begin{gather}
		Y=\bar{\Theta} X + \bar{W}, \label{eq:Param:Ass:Dyn}\\
		\begin{bmatrix}	\bar{W}\\ I	\end{bmatrix}^\top
		\Phi
		\begin{bmatrix}	\bar{W}\\ I	\end{bmatrix}\succ 0. \label{eq:Param:Ass}
	\end{gather}
\end{assumption}

This is a Slater condition and ensures the existence of a relative interior of $\Sigma_\Theta$. 
Note, the existence of a relative interior of $\mathcal{W}$ does not suffice to satisfy Assumption\,\ref{ass:param:Slater}, as \eqref{eq:Param:Ass:Dyn} may intersect $\mathcal{W}$ only at the boundary, guaranteeing only satisfaction with $\succeq$.
If needed, strict satisfaction is achieved by replacing $\Phi$ with $\Phi+\mathrm{diag}(0,\epsilon I)$ and $\epsilon>0 $ sufficiently small.
 Using this assumption, we can state the main theorem of the paper.

\begin{theorem} \label{theo:Param:Consistency}
	Suppose Assumption\,\ref{ass:param:Slater} is satisfied and there exists $\tilde{X}\in\mathbb{R}^{(N-n)\times N}$, $\Theta_0\in\mathbb{R}^{p \times n}$, $\tilde{\Theta}_0\in\mathbb{R}^{p \times (N-n)}$, $M\in\mathbb{R}^{(N-n)\times (N-n)}$, $G\in\mathbb{R}^{N\times n}$, and $\tilde{G}\in\mathrm{R}^{N\times (N-n)}$ with
	\begin{align}
		Y &= \begin{bmatrix} \Theta_0& \tilde{\Theta}_0	\end{bmatrix} \begin{bmatrix}X \\ \tilde{X}	\end{bmatrix} \label{eq:Param:Splitting}		
	\end{align}
	and 
	\begin{align}
		& I =\begin{bmatrix}X \\ \tilde{X}\end{bmatrix}	\begin{bmatrix}G & \tilde{G}\end{bmatrix}, & 
		 M = \begin{bmatrix} \tilde{\Theta}_0 \\ \tilde{G}\end{bmatrix}^\top \Phi \begin{bmatrix} \tilde{\Theta}_0 \\ \tilde{G}\end{bmatrix}.\label{eq:Theor:Space} 
	\end{align}
	Then it holds
	\begin{align*}
	\Sigma_\Theta = \big\{ \Theta\in\mathbb{R}^{p\times n} \mid  \eqref{eq:Param:TheoConsistQMI} \text{ is satisfied.}
	\big\}.
	\setcounter{equation}{\value{equation}+1}
	\end{align*}
\end{theorem}
\begin{proof}
	First, we show that Assumption\,\ref{ass:param:Slater} implies $M\succ 0$.
	To this end, we insert \eqref{eq:Param:Ass:Dyn} and \eqref{eq:Param:Splitting} into \eqref{eq:Param:Ass} to get
	\begin{equation*}
		\begin{bmatrix}	\begin{bmatrix} \Theta_0-\bar{\Theta} & \tilde{\Theta}_0 \end{bmatrix} \begin{bmatrix} X \\ \tilde{X} \end{bmatrix}\\ I	\end{bmatrix}^\top
		\Phi
		\begin{bmatrix}	\begin{bmatrix} \Theta_0-\bar{\Theta} & \tilde{\Theta}_0 \end{bmatrix} \begin{bmatrix} X \\ \tilde{X} \end{bmatrix}\\ I	\end{bmatrix}
		\succ 0.
	\end{equation*}
	By construction, $\tilde{G}$ has full column rank, hence we pre- and post-multiply with $\tilde{G}^\top$ and $\tilde{G}$, which yields $M\succ 0$.
	
	Applying the same steps for a general $\Theta\in\Sigma_\Theta$, but performing a congruence transformation with $[G\;\tilde{G}]$ yields
	\begin{equation*}
		\begin{bmatrix}
			\begin{bmatrix}\hat{\Theta} \\ G\end{bmatrix}^\top \Phi \begin{bmatrix}\hat{\Theta} \\ G\end{bmatrix} & 
			\begin{bmatrix}\hat{\Theta} \\ G\end{bmatrix}^\top \Phi \begin{bmatrix}\tilde{\Theta}_0 \\ \tilde{G}\end{bmatrix} \\
			\begin{bmatrix}\tilde{\Theta}_0 \\ \tilde{G}\end{bmatrix}^\top \Phi \begin{bmatrix}\hat{\Theta} \\ G\end{bmatrix} &
			\begin{bmatrix}\tilde{\Theta}_0 \\ \tilde{G}\end{bmatrix}^\top \Phi \begin{bmatrix}\tilde{\Theta}_0 \\ \tilde{G}\end{bmatrix}
		\end{bmatrix} \succeq 0
	\end{equation*} 
	with $\hat{\Theta}=\Theta_0- \Theta$.
	By Assumption\,\ref{ass:param:Slater}, we know that the lower right block is positive definite and hence invertible, which allows us to employ the Schur complement to get
	\begin{equation*}
		\begin{bmatrix}\hat{\Theta} \\ G\end{bmatrix}^\top \Phi \begin{bmatrix}\hat{\Theta} \\ G\end{bmatrix} - \begin{bmatrix}\hat{\Theta} \\ G\end{bmatrix}^\top \Phi \begin{bmatrix}\tilde{\Theta}_0 \\ \tilde{G}\end{bmatrix} M^{-1}   \begin{bmatrix}\tilde{\Theta}_0 \\ \tilde{G}\end{bmatrix}^\top \Phi \begin{bmatrix}\hat{\Theta} \\ G\end{bmatrix}\succeq 0.
	\end{equation*}
	Sorting all terms and reinserting $\hat{\Theta}=\Theta_0- \Theta$ results in \eqref{eq:Param:TheoConsistQMI}.
	Since all steps hold with equivalence, we can conclude that any $\Theta$ satisfying \eqref{eq:Param:TheoConsistQMI} must also satisfy $\Theta\in\Sigma_\Theta$.
\end{proof}
\begin{figure}[t]
	\begin{center}
		\begin{tikzpicture}[scale = 1.7]
			\draw[->] (-1.05,0) -- (1.2,0) node[below]{$\tilde{X}$};
			\draw[->] (0,-1.2) -- (0,1.2) node[left]{$X$};
			\draw (0,0) circle(1);
			\node at (0.85,0.85){$\mathcal{W}$};
			
			\draw [-] (-0.5,{-sqrt(3/4)}) -- (-0.5,{sqrt(3/4)});
			
			\draw [thick,decoration={brace,	raise=1	},decorate] (0,0)--(-0.5,0);
			\node at (-0.25,-0.2){$\tilde{\Theta}_0 \tilde{X}$};
			
			\draw[dashed] (-0.5,{sqrt(3/4)}) -- (-1.1,{sqrt(3/4)});
			\draw[dashed] (-0.5,{-sqrt(3/4)}) -- (-1.1,{-sqrt(3/4)});
			\draw [thick,decoration={brace, mirror,	raise=1	},decorate] (-1.1,{sqrt(3/4)})--(-1.1,-{sqrt(3/4)});
			\node[rotate=90] at (-1.3,0){$(\Theta_0 -\Theta)X$};
		\end{tikzpicture}
	\end{center}
	\caption{2D-representation for $\Phi = \mathrm{diag}(-I,I)$.}\label{fig:ana:2D-Visualization}
\end{figure}
Using Theorem\,\ref{theo:Param:Consistency}, we can parameterize all $\Theta\in\Sigma_\Theta$ by satisfaction of \eqref{eq:Param:TheoConsistQMI}. 
This is an equivalent description as in \cite{Waarde2022}, but in primal space instead of dual.
Transforming between primal and dual space is possible using dualization, but requires additional invertibility assumptions and is restricted to bounded sets \cite{Mishra2022}.
A simplified visualization of the proof is shown in Fig.\,\ref{fig:ana:2D-Visualization} for $\Phi = \mathrm{diag}(-I,I)$ and interpreting $X$ and $\tilde{X}$ as two basis vectors of a vector space to make up $W$. 
All possible consistent $W$ must satisfy $W = \tilde{\Theta}_0 \tilde{X} + (\Theta_0- \Theta )X\in\mathcal{W}$. 
However, since $\tilde{\Theta}_0$ is fixed and part of every consistent $W$, this restricts all possible $\Theta$.
In contrast, multiplying \eqref{eq:Param:Ass:Dyn} with $G$ to find an expression in terms of $WG$ leads to additional conservatism, since $\tilde{\Theta}_0\tilde{X}G=0$.
Hence, the effects of  $\tilde{\Theta}_0\tilde{X}$ are neglected by projecting on a lower-dimensional subspace. 
Note that \eqref{eq:Theor:Space} implies $X$ has full row rank with $X^\perp$ as a possible choice of $\tilde{X}$ and $G$ being a right inverse of $X$.

\section{Analysis} \label{sec:Analysis}
In this section, we analyze the effects of including the consistency constraint by comparing Theorem\,\ref{theo:Param:Consistency} to \cite{DePersis2020, Berberich2019, Liu2023a} showcasing its usefulness as general analysis tool.
This approach drops the consistency constraint, but adds the right inverse of $X$ as degree of freedom to reduce conservatism.
Hence, we quantify the effects of the right inverse compared to including the consistency constraint.
For simplicity, we assume $\Phi=\mathrm{diag}(-Q,R)$ with $0\prec Q\in\mathbb{R}^{p\times p}$ and $0\prec R\in\mathbb{R}^{N \times N}$. 
We first transform \eqref{eq:Param:TheoConsistQMI} to a block diagonal structure.
\begin{corollary} \label{corr:param:consistency}
	Suppose Assumption\,\ref{ass:param:Slater} is satisfied, $R\succ 0$ and $X$ has full row rank,
	then 
	\begin{align*}
		&\Sigma_\Theta = \left\{ \Theta\in\mathbb{R}^{p\times n} \mid \begin{bmatrix}
			\Theta - \Theta_0\\ I 
		\end{bmatrix}^\top P \begin{bmatrix}
			\Theta - \Theta_0\\ I 
		\end{bmatrix} \succeq 0\right\},
	\end{align*}
	with 
	\begin{equation*}
		P = \begin{bmatrix}
			-Q - Q\tilde{\Theta}_0 M^{-1} \tilde{\Theta}_0^\top Q & 0 \\
			0 & (XR^{-1}X^\top)^{-1}
		\end{bmatrix}
	\end{equation*}
	and
	\begin{align}
		Y= \begin{bmatrix} \Theta_0& \tilde{\Theta}_0	\end{bmatrix} \begin{bmatrix}X \\ X^\perp R	\end{bmatrix}\\
		M = (X^\perp R X^{\perp\top})^{-1} - \tilde{\Theta}_0^\top Q \tilde{\Theta}_0.
	\end{align}
\end{corollary}
\begin{proof}
	This is a special case of Theorem\,\ref{theo:Param:Consistency}, by choosing $\tilde{X}=X^\perp R$, $G=R^{-1}X^\top(XR^{-1}X^\top)^{-1}$ and $\tilde{G}=X^{\perp\top}(X^\perp R X^{\perp\top})^{-1}$ to get $G^\top R \tilde{G}=0$ and $G^\top RG=(XR^{-1}X^\top)^{-1}$.
\end{proof}

Now, we introduce a parameterization, which does not include the consistency constraint.
\begin{theorem}[\hspace{-0.08em}\cite{Berberich2019,Koch2020}]\label{Theo:Ana:InConsistFull}
Suppose $X$ has full row rank.
If there exists a matrix $G\in\mathbb{R}^{N \times n}$ such that
\begin{equation*}
	XG = I,
\end{equation*}
then $\Sigma_\Theta \subseteq \Sigma_\Theta^\mathrm{S,G} \coloneq  \left\{\Theta \mid \Theta = (Y-W)G, \quad W\in\mathcal{W} \right\}$.
\end{theorem}

Note that in comparison to Theorem\,\ref{theo:Param:Consistency}, the consistency constraint $Y=\Theta X + W$ is dropped. 
Instead, by multiplying with the right inverse $G$, the problem is projected on a lower-dimensional subspace.
The right inverse $G$ is a decision variable, which affects the tightness of the overapproximation. 
In Theorem\,\ref{theo:Param:Consistency}, on the other hand, different $\tilde{X}$ do not affect $\Sigma_\Theta$.
For a better comparison, we now transform $\Sigma_\Theta^\mathrm{S,G}$ into a \ac{QMI} representation as in Theorem\,\ref{theo:Param:Consistency} and Corollary\,\ref{corr:param:consistency}.
To do so, we first introduce the following Lemma.
\begin{lemma}[\hspace{-0.04em}\cite{Waarde2023a}] \label{lemma:Ana:GTrick}
	Suppose $R\succ 0$, $W_0\in\mathbb{R}^{p \times N}$ and $G\in\mathbb{R}^{N \times n}$ has full column rank, then
	\begin{align*}
		&\left\{ WG\in\mathbb{R}^{p \times n} \mid \star^\top \begin{bmatrix} - Q & 0 \\ 0 & R	\end{bmatrix} \begin{bmatrix} W-W_0 \\ I\end{bmatrix} \succeq 0 \right\} \\
		= &\left\{\Delta\in\mathbb{R}^{p \times n} \mid \star^\top \begin{bmatrix} - Q & 0 \\ 0 & G^\top R G	\end{bmatrix} \begin{bmatrix} \Delta-W_0G \\ I\end{bmatrix} \succeq 0\right\}.
	\end{align*}
\end{lemma}
\begin{proof}
	In the following, we show equivalence of both parametrizations by considering each set inclusion separately.
	$\subseteq$: This fact follows directly from defining $\Delta = WG$ and pre- and postmultiplying the \ac{QMI} in $W$ with $G^\top$ and $G$. \\
	$\supseteq$: Let $\Delta$ satisfy the corresponding \ac{QMI} and define the following candidate solution
	\begin{equation*}
		\hat{W} = \Delta (G^\top R G)^{-1} G^\top R + W_0R^{-1}G^\perp (G^{\perp\top} R^{-1}G^\perp)^{-1}G^{\perp\top}.
	\end{equation*}
	Clearly, it holds $\hat{W}G = \Delta$.
	Invertibility of $G^\top R G$, and $G^{\perp\top} R^{-1}G^\perp$ follows from $R\succ 0$ and $G$ having full column rank. Furthermore, it holds
	\begin{equation*}
		\begin{bmatrix}
			(G^\top RG)^{-1}G^\top R \\ (G^{\perp\top}R^{-1}G^{\perp})^{-1}G^{\perp\top}
		\end{bmatrix}
		\begin{bmatrix}
			G & R^{-1}G^\perp
		\end{bmatrix} = I.
	\end{equation*}
	Hence, we conclude that $[G\quad R^{-1}G^\perp]$ is invertible.
	Now we can plug $\hat{W}$ into the corresponding \ac{QMI} and perform a congruence transformation with $[G\quad R^{-1}G^\perp]$ to get
	\begin{equation*}
		\begin{bmatrix}
			 G^\top R G-\star^\top Q(\Delta - W_0G)  & 0 \\
			                                                  0 & G^{\perp\top}R^{-1}G^\perp
		\end{bmatrix} \succeq 0,
	\end{equation*}
	which is satisfied by definition of $\Delta$ and $R\succ0$.
\end{proof}

This allows us to represent the uncertainty in terms of the product $WG$ instead of $W$, with $W_0$ representing an offset.
This is summarized in the following theorem.
\begin{theorem}	\label{theo:Ana:Superset}
	Suppose $R\succ 0$ and $X$ has full row rank. If there exists a matrix $G\in\mathbb{R}^{N\times n}$ such that
	\begin{equation*}
		XG = I,
	\end{equation*}
	then
	\begin{align*}
		\Sigma_\Theta^\mathrm{S,G} 	&= \left\{ \Theta \mid \star^\top \begin{bmatrix} -Q & 0 \\ 0 & G^\top RG		\end{bmatrix} \begin{bmatrix} \Theta - YG \\ I\end{bmatrix}\succeq 0\right\}
	\end{align*}
\end{theorem}
\vspace{0.5em}
\begin{proof}
	This follows from Theorem\,\ref{Theo:Ana:InConsistFull} by applying Lemma\,\ref{lemma:Ana:GTrick} on $WG$ for $W=Y-\Theta$ and $W_0=Y$. 
\end{proof}

Since $G$ is a design parameter, it can be chosen as any right inverse. 
In the following, we provide a suitable candidate, which can be computed explicitly.
To do so, we parameterize all solutions of $XG=I$ by a homogeneous and a particular solution $\hat{G}$
\begin{align*}
	G &= \hat{G} + X^{\perp\top}\Lambda \\
	\hat{G} &= R^{-1}X^\top (XR^{-1}X^\top)^{-1}
\end{align*} 
with $\Lambda\in\mathbb{R}^{(N-n)\times n}$ being an arbitrary matrix.
Simple calculations yield
\begin{equation*}
	G^\top R G = \hat{G}^\top R\hat{G} + \Lambda^\top X^{\perp}RX^{\perp\top}\Lambda \succeq \hat{G}^\top R\hat{G},
\end{equation*}
making $\hat{G}$ the optimal choice to minimize the volume described by the corresponding ellipse in $\Theta$.
Note, in general $\Sigma_\Theta^\mathrm{S,\hat{G}} \not\subseteq \Sigma_\Theta^\mathrm{S,G}$ due to the shift $YG$.\\
In the remaining work, we consider only the case $\hat{G}$.
\begin{corollary} \label{corr:ana:seteq}
	Suppose Assumption\,\ref{ass:param:Slater} is satisfied, $R\succ 0$, $X$ has full row rank and
	\begin{align*}
		G &= R^{-1}X^\top (XR^{-1}X^\top)^{-1} \\
		\tilde{G}&=X^{\perp\top}(X^\perp R X^{\perp\top})^{-1} \\
		\tilde{\Theta}_0 &= (Y-W_0)\tilde{G} = 0,
	\end{align*}
	then $\Sigma_\Theta = \Sigma_\Theta^\mathrm{S,G}$.
\end{corollary}
\begin{proof}
	The set equality holds by comparing Corollary\,\ref{corr:param:consistency} and Theorem\,\ref{theo:Ana:Superset}.
\end{proof}
Hence, it suffices to check for $\tilde{\Theta}_0=0$, whether including consistency is beneficial.
If the regressand can be represented using the rows of $X$, i.e., $Y=\bar{\Theta}X$, then consistency does not provide any improvement.
Furthermore, this validates that $G = R^{-1}X^\top (XR^{-1}X^\top)^{-1}$ is a good choice, if no further scheme to compute $G$ is available. 
Since, we derived a condition when the sets are equal, we are now interested in a condition on $Y$, when $\Sigma_\Theta$ will only contain $\Theta_\mathrm{tr}$.
\begin{corollary} \label{corr:ana:convergence}
	Suppose Assumption\,\ref{ass:param:Slater} is satisfied, $R\succ 0$, $Q\succ 0$, $X$ has full row rank and
	\begin{equation*}
		\epsilon I \succeq M \succ 0	\qquad \epsilon > 0,
	\end{equation*}
	with corresponding $Y$, then
	\begin{equation*}
		\lim \limits_{\epsilon \to 0} \Sigma_\Theta = YG
	\end{equation*}
	with
	\begin{align*}
		G = R^{-1}X^\top(XR^{-1}X^\top)^{-1}.
	\end{align*}
\end{corollary}
\begin{proof}
	First, we apply Corollary \ref{corr:param:consistency} to transform the \ac{QMI} to a diagonal structure.
	Since $X$ has full row rank and $R\succ 0$, it follows $XR^{-1}X^\top\succ 0$ and $X^\perp R X^{\perp\top}\succ 0$. 
	Moreover, using 
	\begin{equation*}
		\epsilon I \succeq M = (X^\perp R X^{\perp\top})^{-1} - \tilde{\Theta}_0^\top Q \tilde{\Theta}_0 \succ 0
	\end{equation*}
	with $\epsilon>0$ sufficiently small, we conclude $\tilde{\Theta}_0^\top Q \tilde{\Theta}_0\succ 0$ and $\sigma_{\mathrm{min}}(\tilde{\Theta}_0\tilde{\Theta}_0^\top)>0$.
	Next, we analyze the upper left block
	\begin{align*}
			-Q - Q\tilde{\Theta}_0 &M^{-1} \tilde{\Theta}_0^\top Q \preceq -Q - \frac{1}{\epsilon} Q\tilde{\Theta}_0 \tilde{\Theta}_0^\top Q \\
			&\preceq -(\sigma_\mathrm{min}(Q) + \frac{\sigma_\mathrm{min}(Q)^2}{\epsilon}\sigma_\mathrm{min}(\tilde{\Theta}_0 \tilde{\Theta}_0^\top))I.
	\end{align*}
	Plugging this into the \ac{QMI} for consistency and letting $\epsilon$ approach zero, reduces the set to $\Theta=YG$.
\end{proof}
As demonstrated in Corollaries\,\ref{corr:ana:seteq} and \ref{corr:ana:convergence}, $M$ can be analyzed to quantify the improvement by considering only consistent $W$.
If $\tilde{\Theta}_0=0$ holds, $M$ attains it largest value $(X^\perp R X^{\perp\top})^{-1}$ and $\Sigma_\Theta = \Sigma_\Theta^\mathrm{S,G}$.
If $M$ approaches zero, the corresponding set converges to a single matrix.
This corresponds to the case that the true $W$ approaches the boundary of $\mathcal{W}$ such that only $\Theta_{\mathrm{tr}}$ can explain the data.


\section{ESTIMATOR DESIGN}\label{sec:Estimator}

Now, we demonstrate how to use Theorem\,\ref{theo:Param:Consistency} to design an estimator for a discrete-time \ac{LTI} system of the form
\begin{equation}
	\left[\begin{array}{c}\label{eq:Est:LinSystem}
		x_{k+1}  \\\hline z_{\mathrm{p},k} \\ y_k
	\end{array}\right] = 
	\left[\begin{array}{c|c}
		A_{\mathrm{tr}} & B_{\mathrm{p,tr}} \\ \hline
		C_\mathrm{p} & D_\mathrm{p} \\
		C_{\mathrm{y,tr}} & D_{\mathrm{yp,tr}}  
	\end{array}\right]
	\left[\begin{array}{c}
		x_{k} \\\hline w_{\mathrm{p},k}
	\end{array}\right]
\end{equation}
with $x\in\mathbb{R}^n$ being the state of the system, $w_{\mathrm{p}}\in\mathbb{R}^{n_\mathrm{p}}$ being a performance input, $y\in\mathbb{R}^{p_\mathrm{y}}$ being the only measurements available to the estimator, and $z_{\mathrm{p}}\in\mathbb{R}^{p_\mathrm{p}}$ being the signal to be estimated. 
The system matrices $A_{\mathrm{tr}}\in\mathbb{R}^{n \times n}$, $B_{\mathrm{p,tr}}\in\mathbb{R}^{n \times n_\mathrm{p}}$, $C_{\mathrm{y,tr}}\in\mathbb{R}^{p_\mathrm{y} \times n}$ and $D_{\mathrm{yp,tr}}\in\mathbb{R}^{p_\mathrm{y} \times n_\mathrm{p}}$ are assumed to be unknown. The matrices $C_\mathrm{p}\in\mathbb{R}^{p_\mathrm{p} \times n}$ and $D_\mathrm{p}\in\mathbb{R}^{p_\mathrm{p} \times n_\mathrm{p}}$ are assumed to be known, but the following steps can be extended similarly by deriving a third \ac{QMI} for $C_\mathrm{p}$ and $D_\mathrm{p}$.
As an example, for the state observer synthesis, they are known, i.e., $C_\mathrm{p} = I$ and $D_\mathrm{p}= 0$.
The goal is to synthesize an estimator
\begin{equation}\label{eq:Est:Estimator}
	\left[\begin{array}{c}
		\hat{x}_{k+1}  \\ \hline \hat{z}_{p,k}
	\end{array}\right] = 
	\left[\begin{array}{c|c}
		A_\mathrm{E} & B_\mathrm{E} \\ \hline
		C_\mathrm{E} & D_\mathrm{E}
	\end{array}\right]
	\left[\begin{array}{c}
		\hat{x}_{k} \\\hline y_{k}
	\end{array}\right]
\end{equation}
to minimize the effect of $w_{\mathrm{p}}$ on the estimation error $z_{\mathrm{p}}-\hat{z}_{p}$ by minimizing the $\mathcal{H}_\infty$-norm of the corresponding transfer matrix. 
The estimator state is $\hat{x}\in\mathbb{R}^{n}$ with corresponding system matrices.
Note that the $\mathcal{H}_2$-norm and quadratic performance can be handled similarly \cite{Weiland1994, Scherer2000}.

In this work, we consider a two-stage data collection phase. 
In the first phase, we assume to have access to $N$ measurements of the full state $x$ and $N-1$ measurements of the output $y$ and performance input $w_{\mathrm{p}}$. 
This data collection process is affected by some additive noise $w$ and $v$, which can be rewritten in matrix notation as
\begin{align}
	\left[\begin{array}{c}
		X_+ \\\hline Y
	\end{array}\right] = \left[\begin{array} {c|c}
	A_{\mathrm{tr}} & B_{\mathrm{p,tr}} \\ \hline C_{\mathrm{y,tr}} & D_{\mathrm{yp,tr}}
	\end{array}\right]
	\left[\begin{array}{c}
		X \\ \hline W_{\mathrm{p}}
	\end{array}\right] +
	\left[\begin{array}{c}
		W \\ \hline V
	\end{array}\right]. \label{eq:Estim:MatrixNota}
\end{align}
Furthermore, we assume to have prior knowledge about the noise such that $W\in\mathcal{W}_W$ and $V\in\mathcal{W}_V$ as in Section\,\ref{sec:Param} with corresponding matrices $\Phi_W$ and $\Phi_V$.
In the second phase, we only have access to $y$ and want to employ the synthesized estimator to estimate $z_p$, e.g., the full state $x$.
To do so, we use the collected data of the first phase and apply Theorem\,\ref{theo:Param:Consistency} and Theorem\,\ref{theo:Ana:Superset} twice to get the following descriptions of all system matrices consistent with the data
\begin{gather}
\begin{align*}
	\star^\top &P_1 \left[\begin{array}{c} \left[A-A_0\quad B_\mathrm{p}-B_{\mathrm{0,p}}\right] \\I\end{array}\right] &\succeq 0, \\
	\star^\top &P_2 \left[\begin{array}{c} \left[C_\mathrm{y}-C_\mathrm{0,y}\quad D_\mathrm{yp}-D_\mathrm{0,yp}\right] \\I\end{array}\right] &\succeq 0.
\end{align*}
\end{gather}
This can be rewritten as a \ac{LFT}
\begin{align*}
	\left[\begin{array}{c}
		x_{k+1}  \\ \hline z_k \\ z_{\mathrm{p},k} \\ y_k
	\end{array}\right] &= 
	\left[\begin{array}{c|cc}
		A_0 & \left[I_{n}\;0\right] & B_\mathrm{0,p}\\ \hline
		\begin{bmatrix} I \\ 0	\end{bmatrix} & 0 & \begin{bmatrix} 0 \\ I	\end{bmatrix} \\
		C_\mathrm{p} & 0 & D_\mathrm{p} \\
		C_\mathrm{0,y} & \left[0\;I_{p_y}\right]  & D_\mathrm{0,yp}  
	\end{array}\right]
	\left[\begin{array}{c}
		x_{k} \\ \hline w_k \\ w_{\mathrm{p},k}
	\end{array}\right] \\
	w_k &= \Delta z_k, \qquad \star^\top P \begin{bmatrix}\Delta \\ I\end{bmatrix} \succeq 0,
\end{align*}
with
\begin{equation*}
	P = \star^\top P_1 \begin{bmatrix} I_n & 0 & 0 \\ 0 & 0 & I_{n+n_p} \end{bmatrix}+ \star^\top P_2 \begin{bmatrix} 0 & I_{p_y} & 0 \\ 0 & 0 & I_{n+n_p}\end{bmatrix}.
\end{equation*}
Note that $P_1$ and $P_2$ can each be multiplied by a positive constant. This introduces additional decision variables to reduce conservatism through scaling. From this point on, standard robust control techniques are applied to synthesize an estimator.
\begin{theorem} \label{theo:estimator:EstimatorSyn}
	Suppose a general \ac{LFT} 
		\begin{align*}
			\left[\begin{array}{c}
				x_{k+1}  \\ \hline z_k \\ z_{\mathrm{p},k} \\ y_k
			\end{array}\right] &= 
			\left[\begin{array}{c|cc}
				A   & B_\mathrm{w} & B_\mathrm{p} \\ \hline
				C_\mathrm{w} & D_\mathrm{w} & D_\mathrm{wp} \\
				C_\mathrm{p} & D_\mathrm{pw} & D_\mathrm{p} \\
				C_\mathrm{y} & D_\mathrm{yw} & D_\mathrm{yp}
			\end{array}\right]
			\left[\begin{array}{c}
				x_{k} \\ \hline w_k \\ w_{\mathrm{p},k}
			\end{array}\right] \\
			w_k &= \Delta z_k, \qquad \star^\top P \left[\begin{array}{c}\Delta \\ I\end{array}\right] \succeq 0,
		\end{align*}
		then an estimator as in \eqref{eq:Est:Estimator} achieves an $\mathcal{H}_\infty$-norm smaller than $\gamma$, if the following \ac{QMI} is satisfied
		\begin{align*}
			\star^\top \begin{bmatrix}
				-\mathbf{X} &         0 & 0 &               0 & 0 \\
				          0 & -\gamma I & 0 &               0 & 0 \\
				          0 &         0 & P &               0 & 0 \\ 
				          0 &		  0 & 0 & \mathbf{X}^{-1} & 0 \\
				          0 &         0 & 0 &               0 & \frac{1}{\gamma}I
			\end{bmatrix}
			 \begin{bmatrix}
				I & 0 & 0 \\
				0 & 0 & I \\
				0 & I & 0 \\
				\mathbf{C_1} & \mathbf{D_1} & \mathbf{D_{12}} \\
				\mathbf{A} & \mathbf{B_1} & \mathbf{B_{2}} \\
				\mathbf{C_2} & \mathbf{D_{21}} & \mathbf{D_{2}} \\
		\end{bmatrix} \prec 0			
		\end{align*}
		with $\mathbf{X} \succ 0$ and
		\begin{align*}
			\mathbf{X} &= \begin{bmatrix} Y & Y \\ Y & X	\end{bmatrix}
			&\mathbf{A} &= \begin{bmatrix} YA & YA \\ K & XA+LC_\mathrm{y} \end{bmatrix} \\
			\mathbf{B_1} &= \begin{bmatrix} YB_\mathrm{w} \\ XB_\mathrm{w} + LD_\mathrm{yw}	\end{bmatrix}
			&\mathbf{B_2} &= \begin{bmatrix} YB_\mathrm{p} \\ XB_\mathrm{p} + LD_\mathrm{yp}	\end{bmatrix}\\
			\mathbf{C_1} &= \begin{bmatrix} C_\mathrm{w} & C_\mathrm{w}\end{bmatrix}
			&\mathbf{C_2} &= \begin{bmatrix} C_\mathrm{p}-M & C_\mathrm{p}-NC_\mathrm{y}\end{bmatrix}\\
			\mathbf{D_1} &= D_\mathrm{w}
			&\mathbf{D_{12}} &= D_\mathrm{wp}\\			
			\mathbf{D_{21}} &= D_\mathrm{pw}-ND_\mathrm{yw}
			&\mathbf{D_2} &= D_\mathrm{p}-ND_\mathrm{yp}.
		\end{align*}
		The corresponding estimator can be computed as
		\begin{align*}
			&\left[\begin{array}{cc}
				A_\mathrm{E} & B_\mathrm{E} \\ C_\mathrm{E} &D_\mathrm{E}
			\end{array}\right]=\\
			&\begin{bmatrix}
				Y-X & 0 \\ 0 &I
			\end{bmatrix}^{-1}
			\begin{bmatrix}
				KY^{-1} - XAY^{-1} & L \\
				MY^{-1} & N
			\end{bmatrix}		
			\begin{bmatrix}
				Y^{-1} &0 \\ CY^{-1} & I
			\end{bmatrix}^{-1}.
		\end{align*}
\end{theorem}
\vspace{0.5em}
\begin{proof}
	Since the estimator synthesis is rather well established in literature, and for space reasons, we give only a sketch of the proof. For more details, see \cite{Weiland1994, Scherer2008, Geromel1999}.
	According to \cite[Corollary 3.1]{Holicki2023}, it is possible to determine an upper bound $\gamma$ of the $\mathcal{H}_\infty$-norm, if the following \ac{QMI} 
	\begin{align*}
		\star^\top \begin{bmatrix}
		    -\mathcal{X} &         0 & 0 &           0 & 0 \\
			0            & -\gamma I & 0 &           0 & 0 \\
			0            &         0 & P &           0 & 0 \\
			0            &		   0 & 0 & \mathcal{X} & 0 \\
			0            &         0 & 0 &           0 & \frac{1}{\gamma}I
		\end{bmatrix}
		\begin{bmatrix}
			I & 0 & 0 \\
			0 & 0 & I \\
			0 & I & 0 \\
			C_\mathrm{w} & D_\mathrm{w} & D_\mathrm{wp} \\
			A & B_\mathrm{w} & B_\mathrm{p} \\
			C_\mathrm{p} & D_\mathrm{pw} & D_\mathrm{p} \\
		\end{bmatrix} \prec 0			
	\end{align*}
	is satisfied with $\mathcal{X}\succ0$.
	Next, we apply the following transformation on the closed loop including the estimator to convexify the problem \cite[Theorem 4.2]{Weiland1994}
	\begin{align*}
		&\begin{bmatrix}
			A_\mathrm{E} & B_\mathrm{E} \\
			C_\mathrm{E} & D_\mathrm{E}
		\end{bmatrix}=\\
		&\begin{bmatrix}
			\hat{Y}^{-1}-X & 0 \\ 0 &I
		\end{bmatrix}^{-1}
		\begin{bmatrix}
			\hat{K} - XA\hat{Y} & L \\
			\hat{M} & N
		\end{bmatrix}		
		\begin{bmatrix}
			\hat{Y} &0 \\ C_\mathrm{y}\hat{Y} & I
		\end{bmatrix}^{-1}\\
		&\mathcal{X} = \begin{bmatrix} \hat{Y} & \hat{Y} \\ I & 0\end{bmatrix}^{-1} \begin{bmatrix} I & 0 \\ X & \hat{Y}^{-1}-X\end{bmatrix}.
	\end{align*}
	After applying the congruence transformations  $\mathrm{diag}(\hat{Y}^{-1},I,I,I)$ and $\mathrm{diag}(\hat{Y}^{-1},I)$ and defining the new variables $Y = \hat{Y}^{-1}$, $K=\hat{K}Y$ and $M=\hat{M}Y$, we arrive at the corresponding equations.
\end{proof}

Using the Schur complement, the \ac{QMI} can be transformed into an \ac{SDP} involving the unknowns $Y$, $X$, $K$, $L$, $M$, $N$, and any affine parameterization of $P$, while simultaneously minimizing $\gamma$ \cite[Lemma 4.1]{Weiland1994}.
Note that Theorem\,\ref{theo:estimator:EstimatorSyn} requires the uncertainty $\Delta$ to satisfy a \ac{QMI} in $\Delta$ as given in Theorem\,\ref{theo:Param:Consistency}.
We emphasize again that methods based on data informativity lead to a \ac{QMI} in $\Delta^\top$ \cite{Waarde2020a}, resulting in an \ac{SDP} suited for robust state feedback synthesis \cite[Section 8.3]{Weiland1994}.

\section{NUMERICAL EXAMPLE}

In this section, we compare different approaches for solving the estimator synthesis problem using a numerical example. 
We consider a simple fourth-order system with \cite{Berberich2021}
\begin{align*}
	&A_\mathrm{tr} = \begin{bmatrix}
		\!0.921 &     \!\!0 & \!\!0.041 &     \!\!0 \\
		    \!\!0 & \!\!0.918 &    \!\! 0 & \!\!0.033 \\
		    \!\!0 &     \!\!0 & \!\!0.924 &     \!\!0 \\
		   \!\! 0 &     \!\!0 &     \!\!0 & \!\!0.937
	\end{bmatrix}, 
	& C_\mathrm{y,tr} = \begin{bmatrix} \!1 & \!\!0 & \!\!0 & \!\!0\\ \!0 & \!\!1 & \!\!0 & \!\!0	\end{bmatrix},\\		
	&
	B_\mathrm{p,tr} = \begin{bmatrix}
		\!0.017 & \!0.001 & \!0 & \!0 \\
		\!0.001 & \!0.023 & \!0 & \!0 \\
		    \!0 & \!0.061 & \!0 & \!0 \\
		\!0.072 &     \!0 & \!0 & \!0 
	\end{bmatrix},
	&D_\mathrm{yp,tr} = \begin{bmatrix} 0 & \!\! I\end{bmatrix}.
\end{align*}
for which we estimate the true state, i.e., $z_{\mathrm{p},k}=x_k$.
To this end, we generate data $\{x_k,x_{k+1},w_{\mathrm{p},k},y_k\}_{k=0}^{N-1}$ according to \eqref{eq:Est:LinSystem} with $N=100$ by uniformly sampling $100$ times the initial condition from $[-2,2]^4$ and $w_{\mathrm{p},k}\in[-2,2]^4$.
Furthermore, we assume $\{x_{k+1},y_k\}_{k=0}^{N-1}$ is perturbed by additive measurement noise $W\in\mathcal{W}_W$ and $V\in\mathcal{W}_V$ respectively, as in \eqref{eq:Estim:MatrixNota}. The noise terms are parameterized by $\Phi_W=\mathrm{diag}(-I,0.01^2I)$ and $\Phi_V=\mathrm{diag}(-I,0.01^2I)$, i.e., $\sigma_{\mathrm{max}}(W)\leq0.01$, $\sigma_{\mathrm{max}}(V)\leq0.01$ , $x_{k+1}=x_{k+1,\mathrm{tr}}+w_k$ and $y_{k}=y_{k,\mathrm{tr}}+v_k$. 
The index tr denotes the true value.
Next, we split $W$ into two components $W= \tau_1\Theta_W X + \tau_0\tilde{\Theta}_{W,0}X^\perp$ and choose random $\tilde{\Theta}_{W,0}$ and $\Theta_W$ with $\sigma_{\mathrm{max}}(\tilde{\Theta}_{W,0}X^\perp) = 0.01$, $\tau_0\in[0,1)$. 
We choose $\tau_1$ such that $\sigma_{\mathrm{max}}(W)=0.01$.
The term $\tau_0$ allows us to set the $X^\perp$ component of $W$. 
The case $\tau_0=0$ corresponds to Corollary\,\ref{corr:ana:seteq}, for which $W$ affects the data generation in the same way as $X$.
Hence, it is impossible to distinguish between the term $\Theta_{\mathrm{tr}}X$ and $W=\tilde{\Theta}_{W,0}X$, which leads to a larger set $\Sigma_\Theta$.
For $\tau_0\to 1$, this leads to Corollary\,\ref{corr:ana:convergence}, as $\Theta_{\mathrm{tr}}X$ and $W=\Theta_W X^\perp$ form complementary row spaces, such that $W$ can be reconstructed exactly from $Y$ and hence be compensated for.
The noise matrix $V$ is handled similarly.
The separation between perfectly known $x_k$, but noise affected $x_{k+1}$ is done to compare the different methods for different noise while the regressor $[X^\top\;W_\mathrm{p}^\top]^\top$ remains the same.\\
\begin{figure}[t]
	\begin{center}
		\scalebox{0.6}{
%
%
\definecolor{mycolor1}{rgb}{0.85000,0.32500,0.09800}%
\definecolor{mycolor2}{rgb}{0.92900,0.69400,0.12500}%
\begin{tikzpicture}

\begin{axis}[%
width=0.951\linewidth,
height=0.75\linewidth,
at={(0\linewidth,0\linewidth)},
scale only axis,
xmin=0,
xmax=1,
xlabel style={font=\color{white!15!black}},
xlabel={$\tau_0$},
ymin=0,
ylabel style={font=\color{white!15!black}},
ylabel={Relative Error $\epsilon$},
axis background/.style={fill=white},
axis x line*=bottom,
axis y line*=left,
legend style={at={(0.03,0.03)}, anchor=south west, legend cell align=left, align=left, draw=white!15!black}
]
\addplot [color=blue, line width=3.0pt]
  table[row sep=crcr]{%
0	0.104001396380703\\
0.0256407692307692	0.103956871638062\\
0.0512815384615385	0.103823265007669\\
0.0769223076923077	0.103600493260284\\
0.102563076923077	0.103288412650807\\
0.128203846153846	0.102886818194301\\
0.153844615384615	0.102395437397251\\
0.179485384615385	0.101813927907845\\
0.205126153846154	0.101141867195602\\
0.230766923076923	0.100378747677108\\
0.256407692307692	0.0995239591109133\\
0.282048461538462	0.0985767834959877\\
0.307689230769231	0.0975363705395659\\
0.33333	0.0964017262191921\\
0.358970769230769	0.0951716838401632\\
0.384611538461538	0.0938448811431225\\
0.410252307692308	0.0924197249784571\\
0.435893076923077	0.0908943505944379\\
0.461533846153846	0.0892665826715468\\
0.487174615384615	0.0875338646256325\\
0.512815384615385	0.085693200551327\\
0.538456153846154	0.0837410608808905\\
0.564096923076923	0.0816732772385971\\
0.589737692307692	0.0794849057306238\\
0.615378461538462	0.0771700504570847\\
0.641019230769231	0.074721639013438\\
0.66666	0.0721311195019139\\
0.692300769230769	0.0693880528452444\\
0.717941538461538	0.0664795556399435\\
0.743582307692308	0.0633894960897097\\
0.769223076923077	0.0600973404539732\\
0.794863846153846	0.0565763805707271\\
0.820504615384615	0.0527909684447389\\
0.846145384615385	0.0486918651862881\\
0.871786153846154	0.0442079394454852\\
0.897426923076923	0.0392298226595336\\
0.923067692307692	0.0335733493610369\\
0.948708461538462	0.0268797588879185\\
0.95	0.0265033662857397\\
0.96	0.0234126713915997\\
0.97	0.01991919792861\\
0.974349230769231	0.0182248035688908\\
0.98	0.0157998376692823\\
0.99	0.0104793623876061\\
0.995	0.00675519423107888\\
0.999	0.00184094240831189\\
0.9995	0.000687097939412844\\
0.99999	-0.00168935332444921\\
};
\addlegendentry{Theorem 1}

\addplot [color=mycolor1, line width=2.5pt]
  table[row sep=crcr]{%
0	0.104001396380158\\
0.0256407692307692	0.103998511972873\\
0.0512815384615385	0.103989855995869\\
0.0769223076923077	0.103975430487978\\
0.102563076923077	0.10395523372194\\
0.128203846153846	0.103929264757107\\
0.153844615384615	0.103897521591537\\
0.179485384615385	0.103860002629629\\
0.205126153846154	0.103816704387041\\
0.230766923076923	0.103767623457584\\
0.256407692307692	0.103712753178536\\
0.282048461538462	0.103652088013899\\
0.307689230769231	0.103585617591485\\
0.33333	0.10351333125914\\
0.358970769230769	0.10343521444096\\
0.384611538461538	0.103351249246736\\
0.410252307692308	0.103261416138963\\
0.435893076923077	0.103165685480834\\
0.461533846153846	0.103064028558014\\
0.487174615384615	0.102956404811496\\
0.512815384615385	0.102842769481507\\
0.538456153846154	0.10272306610588\\
0.564096923076923	0.102597228635271\\
0.589737692307692	0.102465178139113\\
0.615378461538462	0.102326818081357\\
0.641019230769231	0.102182034434981\\
0.66666	0.1020306865949\\
0.692300769230769	0.101872606408049\\
0.717941538461538	0.101707583819862\\
0.743582307692308	0.10153536086071\\
0.769223076923077	0.101355608351492\\
0.794863846153846	0.10116790086924\\
0.820504615384615	0.100971680011407\\
0.846145384615385	0.100766182349654\\
0.871786153846154	0.100550320935424\\
0.897426923076923	0.100322434711247\\
0.923067692307692	0.10007977238447\\
0.948708461538462	0.0998170617210422\\
0.95	0.0998031282098408\\
0.96	0.0996922280686916\\
0.97	0.0995747338817726\\
0.974349230769231	0.0995208712315679\\
0.98	0.0994474586183976\\
0.99	0.0993022313211865\\
0.995	0.0992141844118909\\
0.999	0.0991164497775736\\
0.9995	0.0990968212021198\\
0.99999	0.0990636280850882\\
};
\addlegendentry{Theorem 3}

\addplot [color=mycolor2, dashed, line width=2.5pt]
  table[row sep=crcr]{%
0	0.104002080453006\\
0.0256407692307692	0.103957554980007\\
0.0512815384615385	0.103823948500215\\
0.0769223076923077	0.103601175368186\\
0.102563076923077	0.103289094670349\\
0.128203846153846	0.10288749919348\\
0.153844615384615	0.102396117302257\\
0.179485384615385	0.10181460701629\\
0.205126153846154	0.101142545973233\\
0.230766923076923	0.100379425392238\\
0.256407692307692	0.0995246358581725\\
0.282048461538462	0.0985774587453708\\
0.307689230769231	0.0975370449534738\\
0.33333	0.0964023994717645\\
0.358970769230769	0.0951723552530225\\
0.384611538461538	0.0938455499692218\\
0.410252307692308	0.0924203911711429\\
0.435893076923077	0.0908950171138837\\
0.461533846153846	0.0892672468578515\\
0.487174615384615	0.0875345266418006\\
0.512815384615385	0.0856938601362237\\
0.538456153846154	0.0837417182074117\\
0.564096923076923	0.081673932228997\\
0.589737692307692	0.0794855573123317\\
0.615378461538462	0.07717069974181\\
0.641019230769231	0.0747222847477842\\
0.66666	0.0721317610656557\\
0.692300769230769	0.0693886914548028\\
0.717941538461538	0.0664801914311904\\
0.743582307692308	0.0633901284722687\\
0.769223076923077	0.0600979684554905\\
0.794863846153846	0.0565770053960025\\
0.820504615384615	0.0527915878478955\\
0.846145384615385	0.0486924794730063\\
0.871786153846154	0.0442085493159091\\
0.897426923076923	0.0392304271592537\\
0.923067692307692	0.0335739463234867\\
0.948708461538462	0.0268803481771068\\
0.95	0.026503955775683\\
0.96	0.0234132565366521\\
0.97	0.0199197793713722\\
0.974349230769231	0.018225382521932\\
0.98	0.0158004128698715\\
0.99	0.0104799329253677\\
0.995	0.00675575986198871\\
0.999	0.00184192560149751\\
0.9995	0.000688840537002847\\
0.99999	-0.00167740875225576\\
};
\addlegendentry{Informativity}

\end{axis}
\end{tikzpicture}%
}
		\caption{Relative Error $\epsilon$ for different $\tau_0$.}
		\label{fig:exa:ErrorRatio}
	\end{center}
\end{figure}
We compare three different approaches. 
The first one employs our new parameterization described in Theorem\,\ref{theo:Param:Consistency}. 
The second uses Theorem\,\ref{theo:Ana:Superset} to represent \cite{Berberich2019} with $G$ from Corollary\,\ref{corr:ana:seteq}.
Last, we employ the data informativity framework described in \cite{Waarde2022}. 
To do so, we first use the Dualization Lemma in \cite{Weiland1994} to transform the \acp{QMI} in $W$ and $V$ into \acp{QMI} of $W^\top$ and $V^\top$.
By applying the same procedure as in \cite{Waarde2022}, we get a consistent parameterization in $[A \; B_\mathrm{p}]^\top$ and $[C_\mathrm{y}\; D_\mathrm{yp}]^\top$. 
Applying the Dualization Lemma a second time yields two \acp{QMI} in $[A\;B_\mathrm{p}]$ and $[C_\mathrm{y}\; D_\mathrm{yp}]$ on which we apply Theorem\,\ref{theo:estimator:EstimatorSyn}.
In Fig.\,\ref{fig:exa:ErrorRatio}, we illustrate the relative error $\epsilon=\frac{\gamma-\gamma_{\mathrm{tr}}}{\gamma_{\mathrm{tr}}}$ averaged over all initial conditions for the achieved upper bound of the $\mathcal{H}_\infty$-estimator synthesis with respect to the synthesis for the true system $\gamma_{\mathrm{tr}}$.

As expected from Corollary\,\ref{corr:ana:seteq}, for $\tau_0=0$ the error ratio $\epsilon$ is the largest.
In contrast, for $\tau_0\to1$, the relative error $\epsilon$ approaches zero when considering Theorem\,\ref{theo:Param:Consistency} and the informativity framework.
Since, the invertibility conditions for dualization are met, both approaches yield an equivalent representation.
Note that the approach based on Theorem\,\ref{theo:Ana:Superset} performs the worst due to not including consistency.

\section{CONCLUSION}
In this work, we introduced a new data-driven parameterization to describe the set of all parameter matrices consistent with noise perturbed data and demonstrated how to use it to synthesize an estimator with performance guarantees.
In contrast to existing consistent parameterizations, this approach does not require additional invertibility assumptions.
Furthermore, we employed the parameterization to derive verifiable conditions on the data, that determine when including consistency leads to a tighter set description and when it does not provide any additional benefit.

\bibliographystyle{IEEEtran}
\bibliography{2025_LCSS_Literature.bib}

\end{document}